\documentclass[runningheads,pdftex]{llncs}
\usepackage[utf8]{inputenc}
\usepackage[T1]{fontenc}
\usepackage{lmodern}
\usepackage{amsmath}
\usepackage{hyperref}
\usepackage{cite}
\usepackage{paralist}
\usepackage{array}
\usepackage{longtable}\LTpre=0pt\LTpost=0pt

\usepackage{amssymb}
\usepackage{bussproofs}
\usepackage{stmaryrd}

\allowdisplaybreaks

\renewcommand{\orcidID}[1]{}

\renewcommand{\orcidID}[1]{}%Do not include them for now
\author{Michael Roberts\inst{1,2}\orcidID{0000-0003-3155-3516}
\and Alexei Kopylov\inst{1}\orcidID{0000-0001-8943-4435}
\and Aleksey Nogin\inst{1}\orcidID{0000-0002-0795-3694}}
\institute{
HRL Laboratories, LLC, Malibu, CA\\
\email{\{mroberts, akopylov, anogin\}@hrl.com}\\ 
\url{https://csrs.hrl.com/}
\and
Cornell University, Ithaca, NY
}

\title{Semantics and Axiomatization for Stochastic Differential Dynamic Logic}
\begin{document}

\newcommand{\ang}[1]{\langle #1 \rangle}
\newcommand{\Rbot}{\mathbb{R}_\bot}
\newcommand{\Val}{\textnormal{Val}}
\newcommand{\boldx}{\mathbf{x}}
\newcommand{\Rstate}{\hat{\mathbb{R}_\bot^V}}
\newcommand{\sem}[1]{\llbracket #1 \rrbracket}
\newcommand{\ifthenelse}[3]{\textnormal{if }#1 \textnormal{ then }#2 \textnormal{ else }#3}
\newcommand{\letin}[1]{\textnormal{ let } #1 \textnormal{ in }}
\newcommand{\suref}{\textnormal{sure}}
\newcommand{\eqf}{\textnormal{eq}}
\newcommand{\deff}{\textnormal{def}}
\newcommand{\crashf}{\textnormal{crash}}
\newcommand{\tail}{\textnormal{tail}}
\newcommand{\Can}{\mathit{Can}}

\newcommand{\progskip}{\textbf{skip}}
\newcommand{\progfail}{\textbf{fail}}
\newcommand{\xsubsubsection}[1]{\vspace{-1.2\baselineskip}\subsubsection*{{#1}}}

\renewcommand{\doi}[1]{DOI: \href{https://doi.org/#1}{\detokenize{#1}}}

% Swap them from normal LLNCS style to aboid confusing -- with -
\renewcommand\labelitemii{\normalfont\bfseries --}
\renewcommand\labelitemi{$\bullet$}

\maketitle

\begin{abstract}
Building on previous work by Andr{\'e} Platzer, we present a formal language for Stochastic Differential Dynamic Logic, and define its semantics, axioms and inference rules. Compared to the previous effort, our account of the Stochastic Differential Dynamic Logic follows closer to and is more compatible with the traditional account of the regular Differential Dynamic Logic. We resolve an issue with the well-definedness of the original work's semantics, while showing how to make the logic more expressive by incorporating nondeterministic choice, definite descriptions and differential terms. Definite descriptions necessitate using a three-valued truth semantics. We also give the first Uniform Substitution calculus for Stochastic Differential Dynamic Logic, making it more practical to implement in proof assistants.
\end{abstract}

\keywords{
Stochastic reasoning \and
Dynamic logic \and
Proof calculus \and
Hybrid systems \and
Theorem proving
}

\section{Introduction}
It is well known that safety of complex hybrid systems, such as cyber-physical systems (whether autonomous or not), cannot be achieved with just simulation and testing~\cite{DBLP:conf/rss/MitschGP13,DBLP:conf/fm/LoosPN11}. The space of possible behaviors is so big that testing and simulation cannot provide sufficient coverage. Achieving high confidence in correctness requires the ability to model the system mathematically and to prove its properties with an aid of an automated reasoning system. Moreover, cyber-physical systems operate in uncertain environments and even modeling such system is a nontrivial task. Thus, we need a system that is able to reason about properties that incorporate such uncertainties. 

Differential Dynamic Logic (dDL) has proven a useful tool for certifying hybrid systems~\cite{Pla17,Pla18}, with a practical implementation in the KeYmaera theorem prover\cite{PQ08}. This is a logic in the style of Propositional Dynamic Logic \cite{PDL}, with the addition of programs $x' = \theta dt ~\&~H$ that allow the state to evolve continuously according to a differential equation $x'=\theta$ for some non-deterministic amount of time, as long as boundary condition $H$ is satisfied. Part of the reason that dDL has been successful in practice is the substantial amount of work done on its theory since it was first proposed. Notably, uniform substitution-based reasoning~\cite{Pla17} allowed a more concise axiomatization of dDL and enabled the move from KeYmaera to KeYmaera~X~\cite{FMQ+15} with a much smaller trusted code base. The same paper introduced differential forms to the calculus, a syntactic way of reasoning about the derivatives of the continuous dynamics, instead of moving them into side-conditions. 
More recently, the introduction of definite descriptions in $dL_i$~\cite{dLi} allowed for reasoning about terms of the form ``the unique $x$ such that $P$''.  This provides a way to reason about terms that may not be defined everywhere, but are necessary in practice, such as square roots. 

dDL and its simplest probabilistic extensions can only reason about those systems whose continuous behavior is fully deterministic. However, many hybrid systems are best modeled using continuous stochastic processes. This may be because they are deployed in a setting where the underlying dynamics are stochastic, such as in processes interacting with physical materials with stochastic properties, or with financial markets --- or because they represent a controller acting under measurement uncertainty. Reasoning about such systems in a dDL style was formulated in Stochastic Differential Dynamic Logic~\cite{Pla11,PWZZ15}. Here, the continuous programs are generalized to stochastic differential equations of the form $x' = \theta dt + \sigma dW$, expressing the change of $x$ in time as depending on not only $\theta$, but also on $\sigma$ and some underlying continuous stochastic process $W$.

In this work, we seek to similarly develop the foundational theory required for a practical implementation of stochastic differential dynamic logic, by introducing a formalization with definite descriptions and differential forms in the uniform substitution style. In addition, unlike the original sDL~\cite{Pla11,Pla11tech_report}, we also allow for programs with true non-determinism, which is important for reasoning about hybrid systems whose design is not fully specified. 

Defining semantics for stochastic differential dynamic logic is a non-trivial task. In fact, as we point out in Section~\ref{sec:sdl-error},
there appears to be an error in the proof that the original sDL semantics is well-defined, in the step establishing that the semantics of formulas is always measurable.
To resolve this, our semantics differs from those of most dDL--style logics in that for a given interpretation of our language, the ``continuous'' programs $x' = \theta ~\&~H$ can only terminate at pre-chosen stopping times, rather than non-deterministically at any point while $H$ is true. A formula that depends on such programs is then judged valid only if there exists a pre-chosen set of stop times that validate it, recovering the non-determinism.

Our semantics differs from those of $dL_i$ \cite{dLi} in another key way: we define formulas to be indeterminate in the case of program failure, and interpret program modalities to quantify over failures, while $dL_i$ ignores them. So we would evaluate $[x:=1 \cup \progfail]x=1$ to be indeterminate, while the original $dL_i$ would evaluate it to true.  By formulating a nondeterministic controller as a nondeterministic guarded choice operator which would $\progfail$ when all guards are simultaneously false, we avoid the common challenge encountered by KeYmaera X users (particularly novices), where it is easy to accidentally state and prove a safety lemma that is vacuously true --- that is, true not because the controller would always keep the system safe, but because the controller definition accidentally excludes the unsafe trajectories from consideration.

This paper is structured as follows. We first briefly review definite descriptions, stochastic processes and stochastic differential equations in Section~\ref{sec:feat}. Next we present the syntax of our SDL formulation, with a brief outline of the intended semantics in Section~\ref{sec:syntax}. We then present the full formal semantics in Section~\ref{sec:semantics}. We define the semantics of our validity judgment, and present some proof rules in Section~\ref{sec:pathwise}. We present an axiomatization of our SDL in Section~\ref{sec:axioms}. We extend our validity judgment to statements about probabilities, and present proof rules and axioms for probabilistic statements in Section~\ref{sec:distribution}. We then outline a uniform substitution calculus for SDL in Section~\ref{sec:us}. We conclude and discuss the next steps in Section~\ref{sec:conclusion}. The proofs of correctness are presented in Appendix~\ref{sec:appendix}.

\section{Background}
\label{sec:feat}
\subsection{Definite Descriptions and Three-Valued Logic}
\textit{Definite descriptions} \cite{dLi} provides $dDL$ with a means of reasoning about non-polynomial functions, and those which may not be everywhere-defined. Here we review that construction, modifying the syntax slightly in our presentation, to reflect the fact that since our semantics will be singleton-valued, we must specify which element of a vector we select instead of returning the whole vector. A definite description is a term $\iota i \phi_d$, where $\phi_d$ is some logical formula that depends on $d$ special variables $\diamond_n$ for $n\in [1..d],~i\in[1..d]$. If there exists a unique assignment to the variables $\diamond_n$, then $\iota i \phi_d$ should have the value of $\diamond_i$ in that assignment. Otherwise, it should fail to denote, or evaluate to some special symbol for failure. For example, $\iota 1 \  {\diamond_1} * {\diamond_1} = y \wedge \diamond_1 \geq 0$ encodes the positive square root of $y$. \cite{dLi} shows encodings for various other useful functions, including trigonometrics and absolute values.

As \cite{dLi} points out, the possibility that terms fail to denote suggests that we should use a three-valued truth semantics. For instance, consider $[y = -1](\iota 1 \ {\diamond_1} * {\diamond_1} = y \wedge \diamond_1 \geq 0 )< 1$, ``the positive square root of -1 is less than 1''. It is clearly not ``true'', but neither do we wish to judge it ``false'' and say that its negation is ``true''. 
To that end, as in \cite{dLi} we adopt a three-valued   \L{}ukasiewicz logic~\cite{lukasiewicz}, with truth-values $\mathcal{L}:= \{\oplus,\oslash,\ominus\}$ corresponding to ``true, indeterminate, false''. This set is ordered $\ominus < \oslash < \oplus$. We define the negation operator to exchange ``true'' and ``false'', but not act on ``indeterminate'': $\bar{\oplus}:=\ominus,~\bar{\oslash}:=\oslash,~\bar{\ominus}:=\oplus$. We define conjunction as a binary operator on $\mathcal{L}$ that returns $\oplus$ only when both arguments are $\oplus$, $\ominus$ if at least one is $\ominus$, and $\oslash$ otherwise. That is, $a \wedge b = \min (a,b)$. 

\subsection{Stochastic Differential Equations}
We give only a bare-bones background here; for a longer exposition in the context of sDL please see \cite{Pla11}, or \cite{Oks92} for an in-depth source.

For some sample space $\Omega$ and set of times $T$, and some $n$, a stochastic process is a map $X: T \times \Omega \rightarrow \mathbb{R}^n$ such that for all $t\in T$, the partial application $X(t)$ is a random variable. We call the other partial application $X(\omega)$ the {\em path} of $\omega$. A particularly well-studied example is Brownian motion $B$. A version of Brownian Motion is a stochastic process with continuous paths, such that its increments $B(t+\delta)-B(t)$  do not depend on $t$, have mean value zero, and are independent for disjoint periods of time.

This lends itself to modeling noisy dynamics. Consider a variable $x$ changing in time, with rate of change depending on $x$ as well as on some noise. Then we can model the noise on path $\omega$ by $\frac{d}{dt}B(\omega)$. Renaming $B$ to $W$\footnote{for Wiener process, and for consistency with the notation of \cite{Pla11}} and letting $x$ itself be a random variable, we can write $dx = b(x)dt + \sigma(x)dW$, with the understanding that this must hold along every path. We may generalize this to the multidimensional case $d\mathbf{x}=\mathbf{b}dt + \mathbf{\sigma }dW$ for $d\mathbf{x}$ a vector of variables, $W$ a vector of versions of Brownian motion, and $\mathbf{b}$ and $\sigma$ respectively a vector and a matrix of appropriate dimension. Solutions to these kinds of stochastic differential equations are given by the It\^{o} integral.

\section{Syntax}
\label{sec:syntax}
We will develop a dynamic logic for reasoning about first-order programs with non-determinism, probabilistic choice, stochastic differential equations, and definite descriptions. Here we present the syntax, broken into three mutually-defined kinds: terms, representing (attempts at) arithmetic computation in a particular program state; programs, representing state transitions; and formulas that represent truth-value assertions about terms and programs. 

Later in section \ref{sec:us}, we will give a rule for instantiating ``universally quantified'' axioms of our logic to particular cases. Because we do not want the added complexity of quantification over terms, programs, or formulas, our language will include symbols for elements of these kinds whose meaning is open to interpretation, so that we can reason with syntactic substitutions.
\subsection{Terms}

\begin{equation*}
\theta, \kappa ::=~~
	c \mid
	x \mid
	\theta* \kappa \mid
	\theta + \kappa \mid 
	f_d(\theta_1,\theta_2...\theta_d) \mid
	d_t (\theta) \mid
	d_{B,x} (\theta) \mid
	\iota i \phi_d \mid
\end{equation*}
In the above, $c$ is a constant, $x$ a variable, and $f_d$ a symbol for a function of arity $d$,
$\iota i \phi_d$ represents a definite description, where $d$ is a positive integer, $i\in [1..d]$, and $\phi_d$ is a formula with no program or formula symbols as subexpressions, and containing $d$ special variable symbols that are not used in any other context, denoted  $\diamond^{[\phi_d],n}$ for $n \in [1..d]$ and an abstract marker $[\phi_d]$. Call the set of such special variables appearing in $\phi_d$, $\diamond^{\phi_d}$.

Variables come from a countable set $V$ that is closed under sub-scripting by $t$ or $B,x$, where $t$ and $B$ are literals, while $x$ is an arbitrary variable. The term $d_t (\theta)$ is a differential that expresses the rate of change of $\theta$ with respect to time. $d_{B,x}$ is similar, but expresses the rate of change relative to a Brownian motion that is associated to variable $x$. 

\subsection{Stochastic Hybrid Programs}
\begin{equation*}
\begin{split}
\alpha, \beta ::=~&
	x_i:=\theta \mid
	x_i:=* \mid
	d\mathbf{x}=\mathbf{b}dt + \mathbf{\sigma }dW ~\&~ H \mid
	\ifthenelse{H}{\alpha}{\beta} \mid
	\alpha \cup \beta \\
	\mid~&
	\alpha ; \beta \mid
	\alpha^* \mid 
	\progskip \mid
	\progfail \mid 
	\gamma
\end{split}
\end{equation*}
For H a formula containing no program as a subexpression, $\gamma$ a program symbol.
$\mathbf{x}$ is a vector of variables, and $\mathbf{b},\mathbf{\sigma}$ respectively a vector and a square matrix of terms of corresponding dimension. $d\mathbf{x}=\mathbf{b}dt + \mathbf{\sigma }dW ~\&~ H $ evolves the system according to the expressed stochastic differential equation for some length of time, as long as $H$ remains true. $x_i:=*$ draws a new value for $x_i$ uniformly from [0,1].  $ \alpha \cup \beta$ represents the nondeterministic choice between $\alpha$ and $\beta$ (note: in some accounts of dDL, notation $\alpha\mid\beta$ is used instead), $\alpha ; \beta$ represents the sequential execution of $\alpha$ and $\beta$, $\alpha^*$ represents an arbitrary (nondeterministic) number of repetitions of program $\alpha$, $\progskip$ is an empty ``no-op'' program, and $\progfail$ is a ``failure'' program --- once executed, it causes everything to become indeterminate.
\subsection{Formulas}
\begin{equation*}
\phi, \psi ::=~~ 
\theta \geq \kappa \mid
\lnot \phi \mid
\phi \wedge \psi \mid
p_d(\theta_1,\theta_2,...\theta_d)\mid
	\ang{\alpha} \phi \mid 
\suref(\phi)
\end{equation*}

Following the above pattern, $p_d$ is a formula symbol of arity $d$. 
$\suref(\phi)$ is true when $\phi$ is true, and false otherwise (that is, if $\phi$ is false or indeterminate). The program modality $\ang{\alpha} \phi$ gives the maximum value of $\phi$ that could be achieved after running $\alpha$. 

As is standard, we will use $\phi \vee \psi$ as syntactic sugar for $\lnot (\lnot \phi \wedge \lnot \psi)$ and $[\alpha]\phi$ for $\lnot \ang{\alpha} \lnot \phi$. Additionally, we use $\text{ind}(\phi)$ for $\lnot \suref (\phi) \wedge \lnot \suref (\lnot \phi)$. Then we define $\phi_1 \rightarrow \phi_2$ as $\lnot \phi_1 \vee \phi_2 \vee (\text{ind}(\phi_1)\wedge \text{ind}(\phi_2))$, and $\phi_1 \leftrightarrow \phi_2$ as $(\phi_1 \rightarrow \phi_2) \wedge ( \phi_2 \rightarrow \phi_2)$.  
\section{Denotational Semantics}
\label{sec:semantics}
In order to account for partial definitions, we take $\Rbot := \mathbb{R} \cup \bot$ with $\bot$ standing in for ``undefined'', so that $\bot$ added to or multiplied by anything is $\bot$. We consider $\Rbot$ to have the ``extended topology'' generated by the open sets of $\mathbb{R}$ and the singleton $\{\bot\}$, and take the Borel sigma algebra on this topology. 

At any point in a program, only finitely many variables will have been used. Let $\Rstate$ be the set of maps from $V$ to $\Rbot$ that are $\bot$ in all but finitely many places. Alternatively, we may be at a point where the program  has crashed, $\triangledown$, or where it has gone past the bounds of its differential equation, $\triangle$. Call the set of valuations $\Val := \Rstate \cup \{\triangledown,\triangle \}$

As in \cite{Pla11}, we will fix a canonical sample space $\Omega$ and a sigma algebra $\mathcal{F}$ on it. We endow it with a probability measure $\mathcal{P}$ and a family of IID uniform random variables $U:\Omega \rightarrow [0,1]$, as well as a filtration and a family of Brownian motions so that we may interpret It\^{o} integrals \cite{Oks92}. Let a \textit{state} $z$ be a random variable on $\Val$;  $z: ~\Omega \rightarrow \Val$. Let $\mathcal{Z}$ be the set of states. 
 
We give a semantics of non-deterministic choice as maps from adversarial choice sequences to outcomes. These choice sequences are  binary streams that are co-finitely 0; we call this set $\mathcal{C}$ and take the usual tail and head functions. Let $\tail(n,C):=\tail^n(C)$. We will often use a choice sequence to select a natural number, so define $nat(C):\mathcal{C} \rightarrow \mathbb{N},~C\mapsto min \{ p|C(p)=0 \} )$ where $C(p)$ is $p$-th element of $C$, starting the count at $0$.

As pointed out above, our language contains symbols for terms, programs, and formulas that are open to interpretation; our semantics will take as an argument an interpretation $I$ that assigns meanings to these symbols. That is, $If_d: \Rbot^d \rightarrow \Rbot,~Ip_d: \Rbot^d \times \Omega \rightarrow \mathcal{L}$, such that $If_d$ and $Ip_d$ are measurable. $I\gamma:\Val \rightarrow   \Omega \rightarrow \mathcal{C}\rightarrow \Val \times \mathcal{C}$ such that:
\begin{compactenum}
\item At every $v\in\Val$, $\omega\in\Omega$, and $C\in\mathcal{C}$, $\pi_2(I\gamma(v,\omega,C)) = \tail(n,C)$ for some $n$, and if $C'$ agrees with $C$ on the first $n$ elements, then $\pi_1(I\gamma(v,\omega,C')) = \pi_1(I\gamma(v,\omega,C))$ and  $\pi_2(I\gamma(v,\omega,C')) = \tail(n,C')$, and
\item For a state $z\in\Val$ and sequence $C\in\mathcal{C}$, the map $\lambda \omega.\pi_1(I\gamma(z(\omega),\omega,C))$ is measurable.
\end{compactenum}

These last conditions are needed for the semantics of programs to well-defined.

\xsubsubsection{Measurability and Stopping Times}\label{sec:sdl-error}
We would like to use our language to reason about the probabilities of formulas being satisfied. Therefore the intended semantics of formulas must be measurable functions. In order to prove this compositionally, we therefore require that terms and programs also have measurable semantics.
On the other hand, the intended semantics of the formula $\ang{d\mathbf{x}=\mathbf{b}dt + \mathbf{\sigma }dW ~\&~ H }\phi$ is the map $\Omega \rightarrow \mathcal{L}$ that tells us, at each $\omega \in \Omega$, the highest truth value we could achieve by stopping at some time before $H$ ceases to be true. In other words, this is the pointwise supremum of the semantics of $\phi$ over each state we obtain by stopping the program at some time $t$. Unfortunately the pointwise supremum of uncountably many measurable functions is not guaranteed to be measurable \cite{bichteler2002stochastic}.

The original sDL paper \cite{Pla11} claims to give measurable semantics, with proof given in \cite[Appendix A.2]{Pla11tech_report}. To prove memorability of semantics of $\ang{\alpha}\phi$
\footnote{Actually, that paper considers the program modality to apply to real-valued terms instead of formulas, with the semantics of $\sem{\alpha}f$ the supremum of the semantics of $f$ over stop times of $\alpha$. We present this modified to be consistent with our syntax, as it doesn't change nature of the proof.}
, it gives the semantics of $\alpha$ as a pathwise right-continuous process returning valuations: $\sem{\alpha}_t$ and attempts to use this to capture the value of $\sem{\phi}$ in state ${\sem{\alpha}_{t'}}$ for irrational times $t'$ with converging sequences of rational times. Thus they only need take suprema over countably many measurable functions at a time. However, as the semantics of $\phi$ along each path need not be continuous in $\Val$, the supremum over such a sequence is not in general equal to the supremum that considers the irrational time directly. Consider for instance $\ang{x=0;~x':=1 dt~\&~x\leq 3 }x*x=2$. The condition is true only when the program is stopped at exactly time $\sqrt{2}$.

To circumvent this, we directly define the semantics of $\ang{\alpha}\phi$ as returning the pointwise supremum over a countable set of stopping times. This set is a part of interpretation, that is, an interpretation $I$ defines a countable subset of $\mathbb{R}^+$, called $I_{times}$. As we will later show, we will consider a formula valid as as long as there exists some set of stopping times validating it --- recovering the intended semantics of non-determinism. 
We denote the set of possible interpretations $\mathcal{I}$.

\subsection{Definitions}
\subsubsection{Term Semantics}
We give term semantics as $ \sem{\theta}: \mathcal{I} \rightarrow \Rstate \rightarrow  \Rbot$, written $I\boldx\sem{\theta}$. Then we define the function $\hat{I \theta}:\Rstate \rightarrow \Rbot$ as $\lambda \mathbf{x}.I\boldx \sem{\theta}$. 
 This semantics is to be extended to $\mathcal{I} \rightarrow \Val \rightarrow \Rbot$ by defining $I\triangledown \sem{\theta}(\omega) = \bot = I\triangle \sem{\theta}(\omega) = \bot$ for all $\theta, \omega$. This further extends to  $\mathcal{I} \rightarrow \mathcal{Z} \rightarrow \Omega \rightarrow \Rbot$ as $Iz\sem{\theta}(\omega):= Iz(\omega)\sem{\theta}.$ 
The semantics of differential terms is suggested by the It\^{o} formula, see \cite[Chapter~4]{Oks92}.
$$\begin{array}{c}
\begin{array}{rl@{~~~~~~}rl}
I\boldx\sem{c}:= & c 	& I\boldx\sem{\theta * \kappa}:= & I\boldx\sem{\theta}*I\boldx\sem{\kappa} \\
I\boldx\sem{x}:= & \boldx (x) & I\boldx\sem{\theta + \kappa}:= &  I\boldx\sem{\theta}+I\boldx\sem{\kappa}\\
\end{array}\\
I\boldx\sem{f_d(\theta_1...\theta_d)}:=  If_d \left( I\boldx\sem{\theta_1}...I\boldx\sem{\theta_d}\right) 
\end{array}$$
For definite descriptions and differentials, we will first define candidate semantic functions $\Can(I,\theta):\Rstate \rightarrow \Rbot$. If the candidates are not measurable we will discard them in favor of the constant $\bot$ function. 
\begin{align*}
\Can(I,~\iota i \phi_d)(\mathbf{x}):= &
	\begin{cases}	
	\mathbf{y}(\diamond^{[\phi_d],i})& \exists! \mathbf{y}\in \Rstate .~  \boldx^{\diamond^{\phi_d}} = \mathbf{y}^{\diamond^{\phi_d}}, ~ \forall \omega.I\mathbf{y}\sem{\phi_d}(\omega)=\oplus \\
	\bot & \text{else}
	\end{cases} \\
\omit\rlap{\parbox{0.95\textwidth}{Since $\phi_d$ doesn't contain programs or formula symbols, its semantics is constant across $\Omega$.}} \\
\Can(I,~d_t (\theta)(\boldx):= & \left( (\sum_{x \in V}x_t \frac{\partial \hat{I \theta }}{\partial x}) + \frac{1}{2} \sum_{x,y \in V}\frac{\partial ^2 \hat{I \theta }}{\partial x \partial y}(\sum_{j\in V}x_{B,j}y_{B,j}) \right)(\boldx) \\
\Can(I,~d_{B,x}(\theta)(\boldx):= & \left( \sum_{y \in V} \frac{\partial \hat{I \theta }}{\partial y}y_{B,x} \right) (\boldx)
\end{align*} 
Where the last two definitions evaluate to $\bot$ when derivatives are undefined.

Now:
\begin{align*}
I\boldx\sem{\iota i \phi_d}:= &
	\begin{cases}	
	\Can(I,~\iota i \phi_d)(\mathbf{x}) & \Can(I,~\iota i \phi_d) \text{ is measurable}\\
	\bot & \text{else}
	\end{cases} \\
I\boldx\sem{\iota i \phi_d}:= &
	\begin{cases}	
	\Can(I,~d_t \theta)(\mathbf{x}) & \Can(I,~d_t \theta) \text{ is measurable}\\
	\bot & \text{else}
	\end{cases} \\
I\boldx\sem{d_{B,x}\theta}:= &
	\begin{cases}	
	\Can(I,~d_{B,x}\theta)(\mathbf{x}) & \Can(I,~d_{B,x}\theta) \text{ is measurable}\\
	\bot & \text{else}
	\end{cases} 
\end{align*} 
\xsubsubsection{Program Semantics}
$\sem{\alpha}:\mathcal{I} \rightarrow \Val \rightarrow   \Omega \rightarrow \mathcal{C}\rightarrow \Val \times \mathcal{C}$ , written $Iv\sem{\alpha}(\omega,C)$. Then we can extend this to $\sem{\alpha}:\mathcal{I} \rightarrow \mathcal{Z} \rightarrow   \Omega \rightarrow \mathcal{C}\rightarrow \Val \times \mathcal{C}$  by setting $Iz\sem{\alpha}(\omega,C):= I(z(\omega)\sem{\alpha}(\omega,C)$.
 For all $I,\alpha,\omega,C$, $I\triangledown\sem{\alpha}(\omega,C) := \triangledown,C$  and $I\triangle \sem{\alpha}(\omega,C) := \triangle,C$. Thus these cases are ignored below.
\begin{align*}
Iv\sem{x:=\theta}(\omega,C) := & \begin{cases} 
\triangledown, C & Iv\sem{\theta}(\omega) = \bot \\
v[Iv\sem{\theta}(\omega) / x],C & \text{else}
\end{cases}  \\
Iv \sem{x:=*}(\omega,C) := & v[ U(\omega)]/v], C\\&\text{for U a never-before-used random variable} \\
Iv\sem{d \mathbf{x} = b dt + s dW \& ~ H} (\omega,C):=&   \\
&\hspace*{-68pt}\begin{tabular}{l}
let  $\int^t = v[\mathbf{x} \mapsto v(\mathbf{x}) + \int_0^t \sem{b}(\omega) dt + \int_0^t \sem{s}(\omega) dW(\omega)$ \\
else $\triangledown$ if that is undefined because $b$  or $s$ are somewhere $\bot$ \\ 
else $\triangle $ if $I\int^t \sem{H}(\omega) \neq \oplus$ in: \\
$\letin{t = I_{times}(nat(C))}$ \\
$ \begin{cases}
\triangledown \text{ (resp }\triangle) & \exists t'. t'\leq t \wedge \int^{t'}=\triangledown \text{ (resp }\triangle) \\ 
\int^t & \text{else}  \\
\end{cases}, \tail(nat(C),C)$
\end{tabular}\\
Iv  \sem{\ifthenelse{H}{\alpha}{\beta}}(\omega,C) := & \begin{cases}
Iv \sem{\alpha}(\omega,C) & Iz\sem{H}(\omega) = \oplus \\
Iv \sem{\beta}(\omega,C) & Iz\sem{H}(\omega) = \ominus \\
\triangledown,C & Iz\sem{H}(\omega) = \oslash
\end{cases}\\  
Iv\sem{\alpha \cup \beta}(\omega,C):= & \begin{cases}
Iv\sem{\alpha}(\omega,\tail(C)) & \text{head}(C) = 0 \\
Iv\sem{\beta}(\omega,\tail(C)) & \text{head}(C) = 1
\end{cases} \\
Iv\sem{\alpha;\beta}(\omega,C) := & \text{let } (v_\alpha,C_\alpha) = Iv\sem{\alpha}(\omega,C) \text{ in }  I v_\alpha \sem {\beta}(\omega,C_\alpha) \\
Iv\sem{\alpha^*}(\omega,C) :=&  Iv\sem{\alpha^{nat(C)}}(\omega, {tail}(nat(C),C))  \\& \text{  where } \alpha^0 = \progskip \text{ and } \alpha^{k+1} = \alpha;\alpha^k. \\
Iv\sem{\gamma}(\omega,C) := &  (I\gamma)(v,\omega,C) \\ 
I v  \sem{\progskip}(\omega,C) :=& v,C  \\
I v \sem{\progfail}(\omega,C) :=& \triangledown, C
\end{align*}

\xsubsubsection{Formula Semantics}
$ \sem{\phi}: \mathcal{I} \rightarrow \Val \rightarrow \Omega \rightarrow \mathcal{L}$, written $Iv\sem{\phi}(\omega)$. Similar to above we can extend this to
$ \sem{\phi}: \mathcal{I} \rightarrow \mathcal{Z} \rightarrow \Omega \rightarrow \mathcal{L}$, written $Iz\sem{\phi}(\omega)$.
For all $\phi$, we define $I\triangledown \sem{\phi}(\omega) :=I\triangle \sem{\phi}(\omega) := \oslash$. The behavior at all other values is defined below:\vspace{-0.5\baselineskip}
\begin{align*}
Iv\sem{\theta \geq \kappa}(\omega):= & \text{let } a = Iv\sem{\theta}(\omega),~b = Iv\sem{\kappa}(\omega) \text{ in }
	\begin{cases}
	\oplus & a-b \geq 0 \\
	\ominus & a-b < 0 \\
	\oslash & a-b = \bot
	\end{cases} \\[-0.3\baselineskip]
Iv\sem{\lnot \phi}(\omega):= & \bar{Iv\sem{\phi}(\omega)} \\
Iv\sem{\phi \wedge \kappa}(\omega):= & \max (Iv\sem{\phi}(\omega), Iv\sem{\kappa}(\omega)) \\
Iv\sem{p_d(\theta_1...\theta_d)}(\omega) := & Ip_d (Iv\sem{\theta_1}..Iv\sem{\theta_d}))(\omega) \\
Iv\sem{\ang{\alpha}\phi}(\omega):=&  \text{let}v_C = \pi_1 Iv\sem{\alpha}(\omega,C) \text{ in } \sup_{C~st~v_C\neq \triangle} Iv_C \sem{\phi}(\omega) \\
Iv \sem{\suref (\phi)} := & \begin{cases} \oplus & Iv\sem{\phi}(\omega)= \oplus  \\ \ominus & \text{else} \end{cases} 
\end{align*}  

\subsection{Measurability Theorems}
The following theorems are proven in Appendix~\ref{sec:appendix}.

\begin{theorem}[Measurability of Term Semantics]\label{measure-term}
For any term $\phi$, any interpretation $I$, any state $z$, $Iz\sem{\phi}:\Omega \rightarrow \Rbot $ is a measurable function.
\end{theorem}

\begin{theorem}[Measurability of Formula Semantics]\label{measure-formula}
For any formula $\phi$, any interpretation $I$, any state $z$, $Iz\sem{\phi}: \Omega \rightarrow \mathcal{L} $ is a measurable function.
\end{theorem}

\begin{theorem}[Measurability of Determinized Program Semantics]\label{measure-program}
For any formula $\alpha$, any interpretation $I$, any state $z$, and any choice sequence $C$, $Iz\sem{\alpha}(C): \Omega \rightarrow \Val $ is a state (a measurable function).
\end{theorem}

\section{Pathwise Reasoning}
\label{sec:pathwise}
First, we consider reasoning about formulas that are true along every path.
\begin{definition}[Pathwise Validity]
A formula $\phi$ is valid under a countable set of stop-times $T$ if  $\forall I \in \mathcal{I}~st.~I_{times}=T, \forall v \in \Rstate, \forall \omega \in \Omega. Iv\omega\sem{\phi}= \oplus$. Then we write $T\vDash \phi$. If $T$ has the property that $T' \supseteq T \rightarrow T' \vDash \phi$, we write $T\Vvdash \phi$. $\phi$ is pathwise valid if $\exists T~st.~ T\Vvdash \phi$. 
\end{definition}

Where a formula appears on a proof tree, it should be interpreted as the assertion that the formula is pathwise valid.

\subsection{Proof Rules}
Note that that if $\phi_1,~ \phi_2$ are valid, there exist $T_1,~T_2$ st $T_1 \Vvdash \phi_1$, $T_2 \Vvdash \phi_2$, so $T_1 \cup T_2 \Vvdash \phi_1$ and $T_1 \cup T_2 \Vvdash \phi_2$ 

The following proof rules are easily provable valid by inspection of the formula semantics, considering the consequent in the union of the two antecedant's stop times.\vspace{-0.6\baselineskip}
\begin{center}
\AxiomC{$\phi_1$}
\AxiomC{$\phi_2$}
\LeftLabel{AND-ELIM}
\BinaryInfC{$\phi_1 \wedge \phi_2$}
\DisplayProof~~~~~
\AxiomC{$\phi_1$}
\AxiomC{$\phi_1 \rightarrow \phi_2$}
\LeftLabel{MP}
\BinaryInfC{$\phi_2$}
\DisplayProof
\end{center}

As $\phi_1 \leftrightarrow \phi_2$ is valid if and only if they have the same semantics for every $I,v,\omega$, we have a syntactic substitution rule\vspace{-0.3\baselineskip}
\begin{prooftree}
\AxiomC{$\phi_1 \leftrightarrow\phi_2$}
\LeftLabel{IFF-SUB}
\UnaryInfC{$\phi_3 \leftrightarrow \phi_3[\phi_1 / \phi_2]$}
\end{prooftree}
Note that one of the usual proof rules of dynamic logic, the rule ``G'' that says we may derive the validity of $[\alpha]\phi$ from that of $\phi$, is not sound in its usual form here, because $\alpha$ may be a failing program. Instead, let $\crashf (\alpha)$ be the formula $\text{ind}([\alpha]0 \geq 0)$, which has semantics of $\oplus$ if $\alpha$ may transition to $\triangledown$ and $\ominus$ otherwise. 
\begin{prooftree}
\AxiomC{$\phi$}
\LeftLabel{G}\
\UnaryInfC{$\crashf(\alpha)\vee [\alpha]\phi $}
\end{prooftree}
In addition, we will use a uniform substitution proof rule to instantiate axioms under some substitution $\sigma$. We present and justify this rule in Section \ref{sec:us}.

\section{Axioms}
\label{sec:axioms}
We present an axiom schema for pathwise validity. Below, each $\theta,\phi,\alpha$ should be interpreted as a (term, formula, program) symbol of arity 0. Where the soundness proofs follow trivially from the defined semantics, we have omitted them (this is most cases).

\subsection{Pathwise Axioms for Formulas}
We start with the axioms that manipulate formulas. 

\noindent\begin{tabular}{@{}p{0.5\textwidth}@{}p{0.5\textwidth}@{}}
\noindent\textbf{Id}
\begin{compactitem}
\item $\phi \leftrightarrow \phi$
\end{compactitem}
\textbf{Higher Order Equality}
\begin{compactitem}
\item $(\phi_1 \leftrightarrow \phi_2)\leftrightarrow (\phi_2 \leftrightarrow \phi_1)$
\item $(\phi_1\leftrightarrow\phi_2)\leftrightarrow (\lnot \phi_1\leftrightarrow\lnot \phi_2)$
\end{compactitem}&
\textbf{Conjunction}
\begin{compactitem}
\item $(\phi_1 \wedge \phi_2)\leftrightarrow(\phi_2 \wedge \phi_1)$
\item $\phi \leftrightarrow (\phi \wedge \phi)$
\item $((\phi_1 \wedge \phi_2)\wedge \phi_3) \leftrightarrow (\phi_1 \wedge (\phi_2\wedge \phi_3))$
\item $\phi_1 \wedge \phi_2 \rightarrow \phi_1$
\end{compactitem}
\end{tabular}\vspace{-1\baselineskip}

Note that double negation elimination is valid, but the law of the excluded middle is not, unless we are dealing with sure quantities.

\noindent\begin{tabular}{@{}p{0.5\textwidth}@{}p{0.5\textwidth}@{}}
\noindent\textbf{Double Negation Elimination}
\begin{compactitem}
\item $\lnot \lnot \phi \leftrightarrow \phi$
\end{compactitem}&
\noindent\textbf{Excluded Middle of Sureness}
\begin{compactitem}
\item $ \suref (\phi) \vee \lnot \suref (\phi)$
\end{compactitem}
\end{tabular}\vspace{-1\baselineskip}

\noindent\textbf{Sureness}\vspace{-0.5\baselineskip}

\noindent\begin{tabular}{@{}p{0.5\textwidth}@{}p{0.5\textwidth}@{}}
\begin{compactitem}
\item $\suref (\phi) \leftrightarrow \suref(\suref(\phi))$
\item $\suref(\phi) \rightarrow \phi$
\item $\suref(\phi) \rightarrow (\phi_2 \rightarrow \phi_1 \wedge \phi_2)$ 
\end{compactitem}
&
\begin{compactitem}
\item $(\lnot \suref(\phi) \leftrightarrow \lnot \phi \vee \text{ind}(\phi))$
\item $\suref(\phi_1 \wedge \phi_2) \leftrightarrow \suref(\phi_1) \wedge \suref (\phi_2)$
\end{compactitem}
\end{tabular}\vspace{-1\baselineskip}

\subsection{Pathwise Axioms for Terms}
All formulas that are valid in the propositional theory of real closed fields are valid here, so we can include as axioms any axiomatization of that theory. 
Additionally, we have the following axioms for differentiating terms:

	\begin{longtable}{@{}>{$}l<{$}@{~~}>{$}r<{$}}
	d_t c = 0  & (d_t c) \\  
	d_{B,i} c = 0  & (d_B c) \\
	d_t x = x_t  & (d_t x) \\  
	d_{B,i} x = x_{B,i}  & (d_B x) \\
	d_t (f+ g) = d_t f + d_t g& (d_t +) \\
	d_{B,i} (f+ g) = d_{B,i} f + d_{B,i} g& (d_B +) \\
	d_t (f*g) = g * d_t f + f* d_t g + \frac{1}{2}\sum_i d_{B,i} f * d_{B,i} g& (d_t *) \\
	d_{B,i} (f* g) = g* d_{B,i} f + f* d_{B,i} g& (d_B *) \\   	
	\end{longtable}

Finally, we have an axiom for substituting equal terms:
\begin{compactitem}
\item $(\theta_1 = \theta_2)\rightarrow (p_1(\theta_1)\leftrightarrow p_1 (\theta_2))$
\end{compactitem}

\subsection{Pathwise Axioms for Programs}
\noindent\begin{tabular}{@{}p{0.5\textwidth}@{}p{0.5\textwidth}@{}}
\textbf{Skip and Fail}
\begin{compactitem}
\item $\ang{\progskip;\alpha }\phi\leftrightarrow \ang{\alpha }\phi$
\item $\ang{\alpha }\phi \leftrightarrow \ang{\alpha;\progskip }\phi$ 
\item $ \ang{\progskip}\phi \leftrightarrow \phi $
\item $\text{ind}(\ang{\progfail;\alpha}\phi)$
\item $\text{ind}(\ang{\progfail;\alpha}\phi)$
\end{compactitem}&
\textbf{Distributivity}
\begin{compactitem}
\item $(\ang{\alpha }(\phi_1 \vee \phi_2) \leftrightarrow \ang{\alpha }\phi_1 \vee \ang{\alpha}\phi_2  $ 
\end{compactitem}
\textbf{Nondeterministic Choice}
\begin{compactitem}
\item $\ang{\alpha \cup \beta}\phi \leftrightarrow \ang{\alpha}\phi \vee \ang{\beta}\phi$
\end{compactitem}
\end{tabular}\vspace{-1\baselineskip}

\noindent\textbf{Conditionals}
\begin{compactitem}
\item $\ang{\ifthenelse{H}{\alpha}{\beta}}\phi \leftrightarrow (H \wedge (\text{ind}(H)\vee\ang{\alpha}\phi)) \vee (\lnot H \wedge (\text{ind}(H)\vee\ang{\beta}\phi))$
\end{compactitem}

\noindent\begin{tabular}{@{}p{0.6\textwidth}@{}p{0.4\textwidth}@{}}
\textbf{Iteration}
\begin{compactitem}
\item $\ang{\alpha^*}\phi \leftrightarrow \phi \vee \ang{\alpha;\alpha^*}\phi$
\item $[\alpha^*]\suref(\phi \rightarrow [\alpha]\phi) \rightarrow \suref(\phi \rightarrow [\alpha^*]\phi)$
\end{compactitem}&
\textbf{Composition}
\begin{compactitem}
\item $\ang{\alpha;\beta}\phi \leftrightarrow \ang{\alpha}{\ang \beta}\phi $
\end{compactitem}
\end{tabular}\vspace{-1\baselineskip}

The composition one necessitates some justification: 
\begin{lemma}[No Look-Ahead Consumption]\label{no-look}
There exists a natural number $n_C$ such that $\pi_2(I v \sem{\alpha}(\omega,C)) = \tail^{n_C} C$. Furthermore for any $C'$ that agrees with $C$ in the first $n_C$ places, $\pi_1 (I v \sem{\alpha}(\omega,C)) = \pi_1(Iv\sem{\alpha}(\omega,C'))$
\end{lemma} 

\begin{proof}
By structural induction on programs.
\end{proof}
\begin{lemma} [Compositionality of Supremum Semantics]
\label{compositions}
$$Iv \sem{\ang{\alpha;\beta}\phi}(\omega) = Iv\sem{\ang{\alpha}\ang{\beta}\phi}(\omega)$$
\end{lemma}

\begin{proof}
$Iv\sem{\ang{\alpha}\ang{\beta}\phi}(\omega) = \letin{v_1 = \pi_1 (I v \sem{\alpha}(\omega,C_1))}\max_{C_1} I v_1\sem{ \ang{\beta}\phi}(\omega) =\\ ~~ \letin{v_1 = \pi_1 (I v \sem{\alpha}(\omega,C_1))}\\~~~~~ \letin{v_2 =\pi_1 (I v_1 \sem{\beta}(\omega,C_2)) }
\max_{C_1} \max_{C_2}I v_2\sem{ \phi}(\omega)
$.
 By Lemma~\ref{no-look}, this $= \letin{v_C = Iv\sem{\alpha;\beta}(\omega,C)}\max_C  v_C \sem{\phi}(\omega) = Iv\sem{\ang{\alpha;\beta}}(\omega)$
\end{proof}

\noindent\textbf{Assignment}
\begin{compactitem}
\item $\deff (\theta) \rightarrow ([x := \theta]p_1(x) \leftrightarrow p_1(\theta) $ 
\end{compactitem}

\subsection{Differential Axioms}
These axioms correspond to the differential axioms of $dL_i$, figure 4 of \cite{Pla17}. The first two correspond to DW and DC, which say that in the absence of crashes, we may move constraints into postconditions and established postconditions into constraints.
  The next two correspond to DE, and say that after an sde is run, it has performed assignments on the differentials of all involved variables.
   The fifth axiom uses the condition that it have no stochastic behavior to ensure it is exactly axiom DI of \cite{Pla17} --- reducing to the case of ODEs. 

\begin{compactitem}
\item $\crashf (d\mathbf{x}=\mathbf{b}dt + \mathbf{\sigma }dW ~\&~ H) \vee [d\mathbf{x}=\mathbf{b}dt + \mathbf{\sigma }dW ~\&~ H]H$
\end{compactitem}
\begin{compactitem}
\item $\suref([d\mathbf{x}=\mathbf{b}dt + \mathbf{\sigma }dW ~\&~ H_1(x)]H_2(x))\rightarrow \\ \left(\begin{array}{@{}l@{}}[d\mathbf{x}=\mathbf{b}dt + \mathbf{\sigma }dW ~\&~ H_1(x)\wedge H_2(x)]H_3(x)\\\leftrightarrow[d\mathbf{x}=\mathbf{b}dt + \mathbf{\sigma }dW ~\&~ H_1(x)]H_3(x)\end{array}\right)$
\item $[d\mathbf{x}=\mathbf{b}dt + \mathbf{\sigma }dW ~\&~ H]d_t \mathbf{x}_i = \mathbf{b}_i$
\item $[d\mathbf{x}=\mathbf{b}dt + \mathbf{\sigma }dW ~\&~ H]d_{B,x_j} \mathbf{x}_i = \mathbf{\sigma}_{i,j}$
\item $[d\mathbf{x}=\mathbf{b}dt + \mathbf{\sigma }dW ~\&~ H]\mathbf{\sigma} = 0 \wedge d_t\theta_1 \geq d_t \theta_2 \rightarrow ([d\mathbf{x}=\mathbf{b}dt + \mathbf{\sigma }dW ~\&~ H]\theta_1 \geq  \theta_2 \iff (H \rightarrow \theta_1 \geq \theta_2)$
\end{compactitem}

\section{Reasoning in Distribution}
\label{sec:distribution}
We have been careful to establish that our semantics is measurable. Now we can reason about real arithmetic extended with terms of the form $P(\phi)$. As before a set of stop times models a formula $\spadesuit$ of this language, $T\vDash \spadesuit$ when for any $I$ st $I_{times}$ = t, $z$ a state such that $\mathcal{P}(z=\triangle) = \mathcal{P}(z=\triangledown) = 0$, the $\spadesuit$ is a true formula of arithmetic under the substitution $P(\phi) \mapsto  \mathcal{P}(Iz\sem{\phi}(\omega)=\oplus)$. Again we say that a set of stop times validates $\spadesuit$, $T\Vvdash \spadesuit$ when $T' \supseteq T \Rightarrow T'\vDash \spadesuit $, and $\spadesuit$ is valid if $\exists T$ that validates it.  

The following proof rules contain assertions about both the validity of formulas in this arithmetic language, as well as in the language of sDL formulas.

Clearly we can take all the axioms of real arithmetic here, as well the following axioms and proof rules whose soundness is self-evident:\vspace{-0.5\baselineskip}

\noindent\begin{tabular}{@{}p{0.3\textwidth}@{}p{0.35\textwidth}@{}p{0.35\textwidth}@{}}
\begin{compactitem}
\item $ P(\phi)\geq 0$
\item $P(\lnot \phi)\leq 1- P(\phi)$
\end{compactitem}&\begin{compactitem}
\item $P(\phi_1 \vee \phi_2 ) \geq P(\phi_1) $
\item $P(\suref(\phi)) = P(\phi)$
\end{compactitem}&
\vspace{0.01\baselineskip}
\centerAlignProof
\AxiomC{$\phi $}
\LeftLabel{VALID-PROB}
\UnaryInfC{$P(\phi)=1$}
\DisplayProof
\end{tabular}\vspace{-1.75\baselineskip}

\begin{prooftree}
\AxiomC{$(\suref(\phi_1) \rightarrow \lnot \suref (\phi_2) ) \wedge (\suref(\phi_2) \rightarrow \lnot \suref (\phi_1) ) $}
\LeftLabel{DISJOINT-PROB}
\UnaryInfC{$P(\phi_1 \vee \phi_2)=P(\phi_1) + P(\phi_2)$}
\end{prooftree}\vspace{-0.5\baselineskip}

The hardest reasoning principals are those that deal with the program modality. 
The following axiom is valid because every time we randomize  a variable, we do so independently.
\begin{compactitem}
\item $P(\ang{x=*}\phi)=\int_0^1 P(<x=s>\phi~ds$. 
\end{compactitem}
Note that integrals aren't a part of our language, so what we really mean is that the left hand side $=c$ for some constant with the side condition that $c$ = the right hand side. In particular, we can derive from this: 
$0\leq c \leq 1 \rightarrow P(\ang{x=*;\ifthenelse{x<c}{\alpha}{\beta}}\phi>= c* P(\ang{\alpha}\phi) + (1-c)*P(\ang{\beta}\phi) $

\subsection{Inequality Axiom for Stochastic Differential Equations}
Theorem 7 of \cite{Pla11} gives a soundness proof of an axiom for reasoning about stochastic differential equations, drawing heavily on \cite[Section 7.2]{Oks92}. We present the same axiom here, modified to fit our syntax (in particular, note that the extra-syntactic construction $\mathcal{L}$ of \cite{Pla11} is handled by our $d_t$ syntax). We omit the soundness proof, as it is essentially unchanged from \cite{Pla11}.
The following is sound under the following conditions: $\forall I,~ \omega, ~\hat{I\theta}$ has compact real support, and $\mathbf{b},\mathbf{\sigma}$ are real and Lipschitz when restricted to $\{v \in \Val| Iv\sem{H}(\omega)=\oplus\}$. \vspace{-0.7\baselineskip}
\begin{prooftree}
\AxiomC{$\ang{\alpha}H \rightarrow (\theta \leq \lambda p)$}
\AxiomC{$H \rightarrow (\theta \geq 0)$}
\AxiomC{$H\rightarrow (d_t \theta \leq 0)$}
\LeftLabel{$(\ang{'})$}
\TrinaryInfC{$P(\ang{\alpha}\ang{d\mathbf{x}=\mathbf{b}dt + \mathbf{\sigma} dW~\& ~H}\theta \geq \lambda)\leq p$}
\end{prooftree}\vspace{-1\baselineskip}

\section{Uniform Substitution}
\label{sec:us}
In order to make use of our axioms, we need to be able to instantiate them into specific formulas we want to reason about by substituting in concrete terms for function, program, and predicate symbols. A substitution $\sigma$ maps each function symbol $f_d$ to some term $\theta$ containing $d$ special 0-ary function symbols $\bullet_0^{\sigma,f_d,i}$ for $i\in 1...d$, no other 0-ary function symbols, and no variables. Similarly it maps $p_d$ to a formula $\phi$ containing $d$ special 0-ary function symbols $\bullet_0^{\sigma,\phi_d,i}$ for $i\in 1...d$, no other 0-ary function symbols, and no variables. It maps $\gamma$ to some program $\alpha$.  We specify a substitution by a list of mappings, with all unlisted symbols mapped to themselves (in the case of function or predicate symbols, themselves applied to their appropriate $\bullet$ symbols). Define the signature of a substitution $\Sigma(\sigma):= $\{symbols that $\sigma$ does not map to themselves\}.

$\sigma$ acts on terms, programs, and formulas (written $[\sigma \tau]$ for $\tau$ a term, program, or formula) by simultaneously and recursively applying all its mappings to produce a new term, program, or formula respectively, as follows: 
	\begin{longtable}[l]{@{}l@{~~~~}r}
	\textbf{Terms} & \\
	$[\sigma c]$ & $c$ \\
	$[\sigma x]$ & $x$ \\
	$[\sigma (\theta * \kappa)]$ & $[\sigma \theta]*[\sigma \kappa]$ \\ 
	$[\sigma (\theta + \kappa)]$ & $[\sigma \theta]+[\sigma \kappa]$ \\ 
	$[\sigma (d_t(\theta))]$ & $d_t([\sigma\theta])$ \\ 
	$[\sigma (d_{B,i}(\theta))]$ & $d_{B,i}([\sigma\theta])$ \\ 
	$[\sigma f_d(\theta_1...\theta_d)]$ & $[\{\forall i \in (1...d) \bullet_0^{\sigma,f_d,i} \mapsto [\sigma \theta_i]\}(\sigma f_d)] $ \\
	$[\sigma (\iota i \phi_d )$ & $\iota i [\sigma \phi_d]$ \\
	\end{longtable}
	\begin{longtable}[l]{@{}l@{~~~~}r}
	\textbf{Programs} & \\
	$[\sigma (x_i := \theta]$ & $x_i :=[\sigma \theta]$ \\
	$[\sigma (x_i := *]$ & $x_i :=* $ \\
	$[\sigma (d\mathbf{x}=\mathbf{b}dt + \mathbf{\sigma }dW ~\&~ H )] $ & $d\mathbf{x}= [\sigma \mathbf{b}] dt + [\sigma \mathbf{\sigma }]dW ~\&~ [\sigma H] ) $ \\
	$[\sigma (\ifthenelse{H}{\alpha}{\beta})] $ & $\ifthenelse{[\sigma H]}{[\sigma \alpha] }{[\sigma \beta]}$ \\
	$[\sigma (\alpha ; \beta)] $ & $[\sigma \alpha];[\sigma \beta]$ \\
	$[\sigma \alpha^*]$ & $[\sigma \alpha] ^ *$ \\
	$[\sigma \gamma] $ & $\sigma \gamma $ \\
	$[\sigma \progfail] $ & $\progfail $\\
		$[\sigma \progskip] $ & $\progskip $\\
	\end{longtable}
	\begin{longtable}[l]{@{}l@{~~~~}r}
	\textbf{Formulas} & \\
	$[\sigma (\theta \geq \kappa)]$ & $[\sigma \theta] \geq [\sigma \kappa]$ \\ 
	$[\sigma (\lnot \phi)]$ & $\lnot [\sigma \phi]$ \\
	$[\sigma (\phi \wedge \psi)]$ & $[\sigma \phi] \wedge [\sigma \psi] $\\
	$[\sigma (p_d(\phi_1,...\phi_d))]$ & $[\{\forall i \in (1...d) \bullet_0^{\sigma,\phi_d,i} \mapsto [\sigma \phi_i]\}(\sigma p_d)]$  \\
	$[\sigma \suref(\phi)]$ & $\suref([\sigma \phi])$ \\
	$[\sigma \ang{\alpha}\phi]$ & $\ang{[\sigma \alpha]}[\sigma \phi]$ 
	\end{longtable}

$\sigma$ also acts on formulas in our probability meta-language; let $[\sigma \spadesuit]$ be $\spadesuit$ with $P([\sigma \phi])$ substituted for $P(\phi)$.

Then instantiating axioms is done via the following rules, with the admissibility side-condition to be specified later: \vspace{-0.75\baselineskip}

\begin{prooftree} 
\AxiomC{$ \phi$}
\LeftLabel{US}
\RightLabel{$\sigma$ is $\phi$-admissible}
\UnaryInfC{$\sigma[\phi]$}
\end{prooftree}\vspace{-1\baselineskip}
\begin{prooftree} 
\AxiomC{$ \spadesuit$}
\LeftLabel{USP}
\RightLabel{$\sigma$ is $\phi$-admissible for all $P(\phi)$ appearing in $\spadesuit$}
\UnaryInfC{[$\sigma \spadesuit]$}
\end{prooftree}\vspace{-0.75\baselineskip}

Uniform substitution was introduced to dDL in \cite{Pla17}, where one of the key insights was that the admissibility of a substitution on an expression can be related to its ``static semantics'' --- which variables it depends on, and which its programs may modify. Throughout this section, we follow the development and strategies of that paper, but apply them to our sDL. 

\xsubsubsection{Read and Write Variables}
 For a value $v\in \Rstate $ and $V' \subset V$, let $v^{V'}$ be its projection onto the subspace without $V'$. Let $v^x$ be shorthand for $v^{\{x\}}$. Let $v(x)$ be the projection onto only $x$.
\begin{definition}
\begin{align*}
	\mathrm{RV}(\theta) := & \{x\in V | \exists I,v_1,v_2.~v_1^x = v_2^x, Iv_1\sem{\theta} \neq Iv_2\sem{\theta} \} \\
	\mathrm{RV}(\alpha) := & \left\{\begin{array}{@{}l@{}}x\in V | \exists I,v_1,v_2,v_1',\omega,C. ~ v_1^x = v_2^x,Iv_1\sem{\alpha}(\omega,C)=v_1',\\~~~(\nexists v_2'.~ v_1'^x = v_2'^x \wedge Iv_2\sem{\alpha}(\omega,C)=v_2' )\end{array}\right\} \\
	\mathrm{RV}(\phi):= &  \{x \in V| \exists I,v_1,v_2,\omega.~ v_1^x = v_2^x, Iv_1\sem{\phi}(\omega) \neq Iv_2\sem{\phi}(\omega)  \} \\
	\mathrm{WV}(\alpha):= & \{ x\in V| \exists I, v_1,v_2,\omega,C.~ Iv_1\sem{\alpha}(\omega,C)=v_2,~v_1(x)\neq v_2(x))  \} \\
	\mathrm{WV}(\phi):= & \cup_{\alpha \text{ a subexpression of }\phi}\mathrm{WV}(\phi) \\
\end{align*}
\end{definition}

\begin{definition}
The signature $\Sigma$ of a term, program, or formula is set of function, predicate, program symbols it contains.
\end{definition}

\subsubsection{Admissibility Condition}
Define the read variables introduced by $\sigma$ to a program or formula $e$, $\mathrm{RV}(\sigma,e):= \cup_{s \in \Sigma(\sigma)\cap \Sigma(e)}\mathrm{RV}([\sigma s])$.

\begin{definition}
$\sigma$ is defined admissible for $e$ when for any subexpression $e'$ of $e$, $\mathrm{RV}(\sigma,e)\cap \mathrm{WV}([\sigma e'])=\emptyset$.
\end{definition}
\xsubsubsection{Syntactic Approximations}
Note that along the lines of the original US paper \cite{Pla17}, we can syntactically compute over-approximations of the sets of read and write variables. Thus in some cases we can prove that substitutions are admissible solely syntactically. 

\subsection{Soundness}

\subsubsection{Adjoint Substitutions}
For $q \in \Rbot^d$, let $I_{\sigma,f_d}^q$ be the interpretation that is the same as I, but with $\bullet_0^{\sigma,f_d,i}$ mapped to $q(i)$ for all $i$, and let $I_{\sigma,p_d}^q$ be the interpretation that is the same as I but with $\bullet_0^{\sigma,p_d,i}$ mapped to $q(i)$. 

\begin{definition}{Substitution Adjoint} \label{substitution_adjoint}
The \textit{adjoint} to substitution $\sigma$ is the operation that maps $I,v$ to the adjoint interpretation $ \sigma(I,v) $.
	\begin{align*}
	\sigma(I,v) (f_d) = \lambda q: \Rbot^d. I_{\sigma,f_d}^q v \sem{\sigma f_d} \\
	\sigma(I,v) (p_d) =  \lambda q: \Rbot^d.\lambda \omega. I_{\sigma,\phi_d}^q v \sem{\sigma p_d} (\omega)\\
	\sigma(I,v) (\gamma) = \lambda v',\omega',C'. Iv'\sem{[\sigma \gamma]}(\omega',C')\\
	\sigma(I,v)_{times} = I_{times}
	\end{align*} 
\end{definition}
Note that $\sigma(I,v)(\gamma)$  does not depend on $v$. The following corollary follows directly from the above definitions.

	\begin{corollary}{(Equal Adjoint Interpretations)} \label{admissible_adjoints}
	If for some $e$ a term,program,or formula symbol $v_1 = v_2 $ on $\mathrm{RV}(\sigma,e)$, then $\sigma(I,v_1)(e) = \sigma(I,v_2)(e)$. Therefore for any program or formula $e$, if $v_1 = v_2$ on $\mathrm{RV}(\sigma,e)$, $e$ has the same semantics under the interpretations $\sigma(I,v_2)$ and $\sigma(I,v_1)$.
	\end{corollary}

The following Lemmas are proven in Appendix~\ref{sec:appendix}.

\begin{lemma}{Uniform substitution for terms.} \label{uniform_terms}
For all $I,v$,\ \ $Iv\sem{[\sigma \theta]} = \sigma(I,v)v \sem{\theta}$
\end{lemma}

\begin{lemma}{Uniform substitution for programs.} \label{uniform_programs}
When $\sigma$ is admissible for $\alpha$, uniform substitution  and its adjoint interpretation have the same semantics for all $I,v,\omega$:
\ \ $Iv\sem{[\sigma \alpha]}(\omega,C) = \sigma(I,v)v \sem{\alpha}(\omega,C)$
\end{lemma}

\begin{lemma}{Uniform substitution for formulas.} \label{uniform_formulas}
When $\sigma$ is admissible for $\phi$, The uniform substitution $\sigma$ and its adjoint interpretation have the same semantics for all $I,v,\omega$:
\ \ $I v\sem{[\sigma \phi]}(\omega) = \sigma(I,v) v \sem{\phi}(\omega)$
\end{lemma}

\xsubsubsection{US, USP}
US is thus sound: $T\Vvdash \phi$ if and only if for any $I$ with times a superset of $T$, for all $v,\omega$, $Iv\sem{\phi}= \oplus$. If $\sigma$ is admissible for $\phi$ then by theorem \ref{uniform_formulas}, for any $I$ with times a superset of T, $Iv\sem{[\sigma]\phi}(\omega)=\sigma(I,v)v\sem{\phi}(\omega)=\oplus$ as $\sigma(I,v)$ has the same times as $I$. We obtain the soundness of USP similarly.

\section{Conclusion}
\label{sec:conclusion}
	We have given a logic for reasoning about stochastic hybrid programs, with a measurable semantics. To make it suitable for implementation in a practical proof-assistant, we have extended it with definite descriptions and differentials, and given it a proof calculus in a uniform substitution style. 
	\subsection{Future Work}
		We have maintained the measurability of our semantics by adopting an all-or-nothing approach with respect to definite descriptions and differentials, where if the resulting semantics isn't measurable we throw it out. It would be interesting to see if we could handle terms more delicately by only demanding that their semantics be restricted to the contexts of the formulas they appear in. Additionally, sufficient conditions for measurability of partial derivatives are given in \cite{MM77}, and it may be fruitful to incorporate them into our semantics.
	
	We've presented just one proof rule for stochastic differential equations, based on prior work from \cite{Pla11} which leverages Doob's Martingale inequality. There is a rich literature on concentration inequalities for martingales~\cite{howard2020timeuniform,aeckerlewillems2019concentration} as well as inequalities on the solutions of SDEs, such as \cite{compare_sde}. We would like to derive sound proof rules based on these methods, and we believe that the differential terms of our languages will allow us to do so in a way that minimizes the need for semantic side-conditions.

\appendix
\makeatletter\def\@seccntformat#1{\appendixname~\thesection\quad}\makeatother
\section{Proofs}
\label{sec:appendix}

\begin{proof}[Theorem \ref{measure-term}]
 We proceed by structural induction. It is enough that the function $\hat{I \theta}$ be measurable; the intended result then follows from the structure of the $\sigma$-algebra on $\Val$, and composing with $z$ as a measurable function.  
 
The semantics of constants and variables satisfies this immediately. Term symbols, differentials, and definite descriptions do so by definition. Addition and multiplication follow from the inductive hypothesis.

\end{proof}

\begin{proof}[Theorems \ref{measure-formula}, \ref{measure-program}]
We prove these theorems simultaneously by induction on the structure of programs and formulas.

\paragraph{Formulas:}
The measurability of inequalities follows from Theorem~\ref{measure-term}. The cases of negations and conjunctions, and sureness follow from the inductive hypothesis. The interpreted formula symbols are measurable by definition. 

For $\ang{\alpha}\phi$, we have $$\text{let}z_C(\omega) = \pi_1 Iz(\omega)\sem{\alpha}(\omega,C) \text{ in } \sup_{C~st.~z_C(\omega) \neq \triangle} Iz_C(\omega) \sem{\phi}(\omega).$$ By $\ref{measure-program}$ and IH, each $z_C(\omega)$ must be measurable, so by IH so is $Iz_C(\omega) \sem{\phi}(\omega)$. The semantics here are then the pointwise supremum of countably many measurable functions, which is measurable~\cite{MM77}. 

\paragraph{Programs:}
Under a particular choice sequence, $\alpha \cup \beta$ and $\alpha^*$ behave as other programs, so we don't need to consider these cases. $\gamma$ preserves measurability by definition. $\progskip$ acts as the identity transform on states, and $\progfail$ outputs a constant function, so these trivially preserve measurability. That $\alpha;\beta$ preserves measurability follows directly from the IH. Random variables and It\^{o} integrals are measurable by definition. The semantics of assignments and conditionals can be rewritten as compositions of the semantics of terms, formulas, and programs with projections, and thus preserve measurability by the appropriate use of IH and \ref{measure-term}.
\end{proof}

\begin{proof}[Proof of Lemmas \ref{uniform_terms}, \ref{uniform_programs}, \ref{uniform_formulas}]
 We place the following well-founded partial order on substitutions: $\sigma_1 \sqsubseteq \sigma_2$ if $\Sigma(\sigma_1) \subset \textnormal{sym}(\sigma_2)$ or if every element of $\Sigma(\sigma_1)$ is of the form $\bullet_0^{\sigma,\tau_d,i}$ for $\tau_d$ some symbol $f_d$ or $p_d$ in $\textnormal{sym}(\sigma_2)$.
\\ The unique least substitution is then the identity substitution. Observe that from definition \ref{substitution_adjoint}, the lemmas always hold when $\sigma$ is the identity substitution.
\\We proceed by mutual structural induction on terms, programs, and formulas, and simultaneously on the ordering on substitutions. Note that semantics where $v = \triangledown$ doesn't change under reinterpretation, so below we just consider the cases when $v\in \Rstate$.

Let us start by considering terms:
\begin{compactitem}
\item The semantics of constants and variables doesn't change under different interpretations, and don't substitutions don't change them.
\item The cases of additions, multiplications, and derivatives follow immediately by IH.
\item As formulas appearing in definite descriptions contain no programs, they have no write variables, so $\sigma$ is admissible for them and this case also follows by IH.

\item $Iv\sem{[\sigma f_d(\theta_1...\theta_d)]}
= Iv\sem{[\{\forall i \in (1...d) \bullet_0^{\sigma,f_d,i} \mapsto [\sigma \theta_i]\}(\sigma f_d)]}(\omega) \\ 
\stackrel{IH}{=}  \{\forall i \in (1...d) \bullet_0^{\sigma,f_d,i} \mapsto [\sigma \theta_i]\}(I,v)v \sem {\sigma f_d}$ 
\\
Note that $\{\forall i \in (1...d) \bullet_0^{\sigma,f_d,i} \mapsto [\sigma \theta_i]\}(I,v)$ acts the same as $I$ except that it interprets each $\bullet_0^{\sigma,f_d,i}$ as $Iv\sem{[\sigma \theta_i]}$.
So by definition, the above expression $= \sigma(I,v) f_d \left( I v \sem{[\sigma \theta_1]}...I v \sem{[\sigma\theta_d]} \right)
\\
\stackrel{IH}{=} \sigma (I,v) f_d  (\sigma (I,v) v \sem{\theta_1}...\sigma (I,v) v \sem{\theta_d})
=\sigma (I,v) v \sem{f_d(\theta_1...\theta_d)}$
\end{compactitem}
Now consider programs:

\begin{compactitem}
\item $Iv\sem{[\sigma ~ x:=\theta]}(\omega,C) = Iv\sem{x:=[\sigma\theta]}(\omega,C) =\\ \begin{cases} 
\triangledown, C & Iv\sem{[\sigma \theta]} = \bot \\
v[Iv\sem{[\sigma \theta]} / x],C & \text{else}
\end{cases} \\
\stackrel{IH}{=} \begin{cases} 
\triangledown, C & \sigma(I,v) v\sem{\theta} = \bot \\
v[\sigma(I,v)v\sem{\theta} / x],C & \text{else}
\end{cases} = \sigma(I,v)v\sem{x:=\theta}(\omega,C) $  
\item Differential equations follow similarly.
\item The semantics of randomization, $\progskip , ~ \progfail$ doesn't change under interpretation.
\item $I v \sem{[\sigma \gamma]}(\omega,C) = \sigma(I,v) I v \sem{\gamma}(\omega,C)$ by definition.
\item $I v \sem{[\sigma ~\alpha;\beta]}(\omega,C)  = I v \sem{[\sigma \alpha];[\sigma\beta]}(\omega,C) =\\ \letin{ (v_\alpha,C_\alpha) = Iv\sem{[\sigma \alpha]}(\omega,C)}  I v_\alpha \sem {[\sigma \beta]}(\omega,C_\alpha) \stackrel{IH}{=} \\
 \letin{ (v_\alpha,C_\alpha) = \sigma(I,v) v\sem{\alpha}(\omega,C)}  \sigma{I,v_\alpha} \sem {\beta}(\omega,C_\alpha) $. 
 By definition, $v = v_\alpha$ on $\mathrm{WV}([\sigma \alpha])^C$, which $\supseteq \mathrm{RV}(\sigma,\alpha;\beta)$ by admissibility, hence $\supseteq \mathrm{RV}(\sigma,\beta)$ by definition. Then by  cor. \ref{admissible_adjoints}, the above = \\
$ \letin{ (v_\alpha,C_\alpha) = \sigma(I,v) v\sem{\alpha}(\omega,C)}  \sigma{I,v} \sem {\beta}(\omega,C_\alpha) = \sigma(I,v)v\sem{\alpha;\beta}(\omega,C)$
\item Conditional, union, star follow similarly from IH. Star requires a nested induction on n in $\alpha^n$. 
\end{compactitem}

And finally formulas:
\begin{compactitem}
\item For inequalities, this follows from Theorem \ref{uniform_terms} and IH.
\item Since the write variables of $\phi$ are a superset of the write variables of any of its subexpressions, $\sigma$ must be admissible for each of its subexpressions. 
Then we can apply IH for the cases of conjunction, negation, and sureness
\item $Iv\sem{[\sigma p_d(\theta_1...\theta_d)]}(\omega) 
= Iv\sem{[\{\forall i \in (1...d) \bullet_0^{\sigma,p_d,i} \mapsto [\sigma \theta_i]\}(\sigma p_d)]}(\omega) \\ 
\stackrel{IH}{=}  \{\forall i \in (1...d) \bullet_0^{\sigma,p_d,i} \mapsto [\sigma \theta_i]\} (I,v) v \sem {\sigma p_d}(\omega)$\\
Note that $\{\forall i \in (1...d) \bullet_0^{\sigma,p_d,i} \mapsto [\sigma \theta_i]\}(I,v)$ acts the same as $I$ except that it interprets each $\bullet_0^{\sigma,p_d,i}$ as $Iv\sem{[\sigma \theta_i]}(\omega)$.
So by definition, the above expression $= \sigma(I,v)p_d (Iv\sem{[\sigma \theta_1]}, ..... Iv\sem{[\sigma \theta_d]})(\omega).
\\
\stackrel{IH}{=} \sigma(I,v)p_d (\sigma(I,v)v\sem{\theta_1}...\sigma(I,v)v\sem{\theta_d})(\omega))
=\sigma(I,v) \sem{p_d(\theta_1...\theta_d)}(\omega)$
\item $Iv\sem{[\sigma\ang{\alpha}\phi]}(\omega) = Iv\sem{\ang{[\sigma\alpha]}[\sigma\phi]}(\omega) =\\ \letin{v_c = \pi_1 Iv\sem{[\sigma \alpha]}(\omega,C)}\sup_{C~st~v_C\neq \triangle} I v_c \sem{[\sigma \phi]}(\omega) \stackrel{IH}{=}\\ \letin {v_C = \pi_1 \sigma(I,v) v \sem{[\sigma \alpha]}(\omega,C) }\sup_{C~st~v_C\neq \triangle} I v_c \sem{[\sigma \phi]}(\omega) \stackrel{IH}{=}\\
\letin{v_C = \pi_1 \sigma(I,v) v \sem{[\sigma \alpha]}(\omega,C)} \sup_{C~st~v_C\neq \triangle} \sigma(I,v_C)v_C \sem{\phi}(\omega) $.\\
Now by definition, $v_C = v$ on $\mathrm{WV}([\sigma \alpha])^C$, which $\supseteq \mathrm{WV}([\sigma \ang{\alpha}\phi])^C$ by definition, which by admissibility $\supseteq \mathrm{RV}(\sigma,[\sigma \ang{\alpha}\phi])$, which again by definition $\supseteq \mathrm{RV}(\sigma,\phi)$. Hence by corollary \ref{substitution_adjoint}, 
the above term =: \\
$\letin{v_C = \pi_1 \sigma(I,v) v \sem{[\sigma \alpha]}(\omega,C)}\\~~~\sup_{C~st~v_C\neq \triangle} \sigma(I,v)v_C \sem{\phi}(\omega)  = \sigma(I,v) v \sem{<\alpha>\phi}(\omega)$
\end{compactitem}
\end{proof}

\bibliography{../sdl}
\bibliographystyle{splncs04} 

\end{document}